\documentclass{article}  

\usepackage{cite,url}

  \usepackage[pdftex]{graphicx}
  \graphicspath{{../fig/}}
  \DeclareGraphicsExtensions{.pdf,.jpeg,.png}

\usepackage[cmex10]{amsmath}
\usepackage{algorithmic}
\usepackage{algorithm}

\usepackage{amssymb}
\usepackage{pslatex}
\usepackage{amsthm}

\newtheorem{claim}{Claim}
\newtheorem{theorem}{Theorem}
\newtheorem{lemma}{Lemma}
\newtheorem{problem}{Problem}
\newtheorem{property}{Requirement}
\newtheorem{definition}{Definition}

\newcommand{\Section}[1]{\section{#1}}
\newcommand{\Subsection}[1]{\subsection{#1}}

\title{\LARGE \bf
Cooperation with Disagreement Correction\\
in the Presence of Communication Failures~\footnote{This paper appears as a technical report~\cite{DBLP:journals/corr/PonceSF14} and as an extended abstract~\cite{PSF14}. This work was partially supported by the European Union's Seventh Programme for research, technological development and demonstration, through project KARYON, under grant agreement No. 288195.}}

\author{Oscar~Morales-Ponce~\thanks{This work was done during the postdoctoral stay of Oscar Morales-Ponce at
Chalmers University of Technology, G\"{o}teborg, Sweden, e-mail:oscarmponce@gmail.com}
, 
        Elad~M.~Schiller~\thanks{Elad M. Schiller is with the Department of Computer Science and Engineering, Chalmers University of Technology,
e-mail:elad@chalmers.se}
, 
        and~Paolo~Falcone~\thanks{Paolo Falcone is with the department of Signals and Systems,
Chalmers University of Technology, G\"{o}teborg, Sweden, e-mail:falcone@chalmers.se}
}

\date{}

\begin{document}
\maketitle
\begin{abstract}
Vehicle-to-vehicle communication is a fundamental requirement for maintaining safety standards in high-performance cooperative vehicular systems. The vehicles periodically exchange critical information among  nearby vehicles and determine their maneuvers according to the information quality and the established strategies. However, wireless communication is failure prone. Thus, participants can be unaware that other participants have not received the needed information on time. This can result in conflicting (unsafe) trajectories. We present a deterministic solution that allows all participants to use a fallback strategy in the presence of communication delays. We base our solution on a timed distributed protocol. In the presence of message omission and delay failures, the protocol disagreement period is bounded by a constant (in the order of milliseconds) that may depend on the message delay in the absence of these failures. We demonstrate correctness and perform experiments to corroborate its efficiency. We explain how vehicular platooning can use the proposed solution for providing high performance while meeting the safety standards in the presence of communication failures. We believe that this work facilitates the implementation of cooperative driving systems that have to deal with inherent (communication) uncertainties.
\end{abstract}


\Section{Introduction}
The vision of automated driving systems holds a promise to change the transportation reality. Current deployments that focus on autonomous solutions pose a variety of sensors and actuators for safe driving on the road, e.g., Volvo \emph{drive me} project in Gothenburg and \emph{Google car} in California. These autonomous solutions are based on the vehicles' ability to observe obstacles in their line-of-sight. Vehicle-to-vehicle communication has the potential to improve the system confidence on the sensory information and support advanced vehicular coordination. E.g., when changing lanes and crossing intersections, as well as improving the road capacity by reducing the inter-vehicle distances. However, communication failures can result in hazardous situations due to coordination based on inconsistent information shared by the participating vehicles.

\begin{figure}[!t]
     \centering
     \includegraphics[scale=0.9]{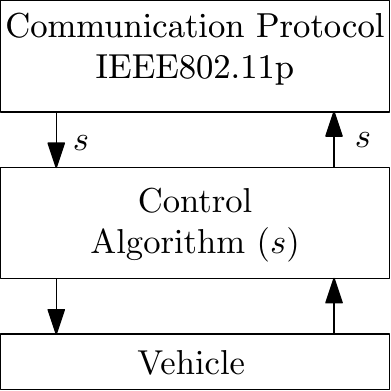}
\hspace{1cm}
     \includegraphics[scale=0.9]{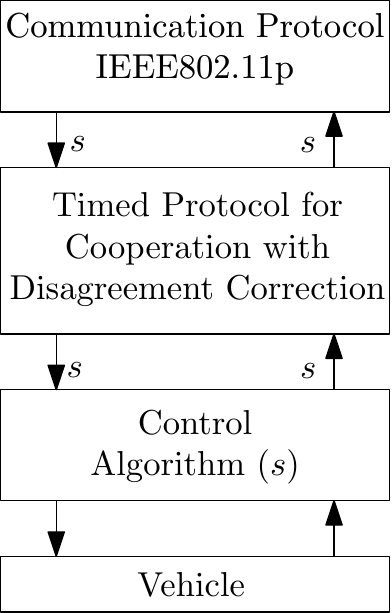}
     \caption{The local cooperative driving architecture without and with the proposed cooperation protocol.}
     \label{fig:architecture}
\end{figure}

Consider an architecture, which Figure~\ref{fig:architecture} (left) depicts, for implementing cooperative driving systems. The Communication Protocol implements the mechanisms for exchanging information with other vehicles. The Control Algorithm plans the vehicle motion according to the sensory information from on-board and remote sources. Note that the local Control Algorithms depend on the (in general vectorial) variable~$s$ (service level). Thus, $s$ is a common piece of information that all vehicles share in order to establish correct cooperation. For instance, in vehicular platooning,~$s$ might include the maximum acceleration levels imposed to all vehicles by the limited braking capabilities of one of them~\cite{6646342}. Clearly, message loss when a new value of~$s$ is established may lead to an inconsistent value in one or more vehicles, and thus, result in an unsafe operation of the entire cooperative system. It is then necessary to have an additional layer, shown in Figure~\ref{fig:architecture} right, between the Communication Layer and the Control Algorithm. We propose to base this additional layer on a Timed Protocol for Cooperation with Disagreement Correction that  resolves disagreements on variable~$s$ among the system vehicles.

Specifically, we address the following research question: How can cooperative systems be used to attain the highest performance without compromising safety in the presence of communication failures? We consider applications in which the individual vehicles estimate their ability to cooperate according to the sensory information quality and communicate their maximum supported cooperative level~\cite{CRP14,DBLP:conf/sss/CasimiroKKSTCPJL12}. The vehicular system then decides on its cooperative service level according only to the received information. However, communication failures can cause the arrival of the needed information not to occur by the deadline. This can bring the vehicles to operate at distinct levels. It is a critical issue to guarantee that the uncertainty period along vehicles occurs only in short time periods. Therefore, we address Problem~\ref{pr:1}.

\begin{problem}[Minimum Longest Uncertainty Period.]
~\label{pr:1}
Is there an upper-bound on the longest period in which the cooperative system may have inconsistent operation service level?
\end{problem}

We note that we cannot solve Problem~\ref{pr:1} using distributed (uniform) consensus algorithms. In the uniform consensus problem, every component (vehicle) proposes a value and the objective is to select exactly on of these proposed values. It is well-known that this problem is not deterministically solvable in unreliable synchronous networks and any $r$-communication rounds algorithm has probability of disagreement of at least $\frac{1}{r+1}$~\cite{DBLP:books/mk/Lynch96} (Theorem 5.1 and Theorem 5.5). Therefore, when the communication failures are too frequent and severe, the uncertainty period cannot be bounded since the components (vehicles) can disagree for an unbounded number of protocol executions. This work presents a communication protocol that guarantees the shortest possible uncertainty period, i.e., a constant time, in the presence of communication failures.

Our solution is based on a communication protocol that collects values from all system components. Once this proposed set $s$ is delivered to all the components, the protocol employs a deterministic function to decide on a single value from $s$ that all system components are to use. The protocol identifies the periods in which there is a clear risk for disagreement due to temporary communication failures, i.e., a period in which $s$ was not delivered by the due time to the entire system. Once such risk is identified, the protocol triggers a correction strategy against the risk of having disagreement for more than a constant number of rounds. Namely, after the occurrence of communication failures that jeopardize safety, all system components will rapidly start a period to reestablish their confidence by returning to the default value. Once the network returns to be stable again, and no communication failures occur, within a constant time, the protocol behaves as if no communication failures has ever occurred.



The correctness proof and its validation show that the proposed solution provides a trade-off between the uncertainty period (in the order of milliseconds) and the occurrence of communication failures. In other words, for shorter round length (and consequently so it the uncertainty period), the vehicles experience more frequently a low service level. However, for a longer round length, the vehicles experience less frequently a low service level. However, the longer the round length is, the longer the time that vehicles spend on disagreements and therefore, the risk of having accidents increases.

This paper also discusses a safety-critical application that facilitates cooperation using the proposed protocol. We assume a baseline adaptive cruise control (ACC) that does not require communication. Then, we extend it to a cooperative one that attains higher vehicle performance, but relies on higher confidence level about the position and velocity of nearby vehicles. We explain how the protocol can provide a timed and distributed mechanism for facilitating decisions about when the vehicles should plan their trajectories according to the baseline application and when according to the extended one that fully utilizes cooperative functionality.

\Subsection{Related work}
The distributed (uniform) consensus problem considers the selection of a single value from a set of values proposed by members of, say, a vehicular system. The solution is required to terminate within a bounded time by which all system components have decided on a common value. The use of the exact (uniform) vs. approximate consensus approaches is explain here~\cite{DBLP:journals/computer/LalaHA91}, where they recommend the use of exact (uniform) consensus due to the simplicity of the system design from the application programmer perspective. The exact consensus approach, in contrast to the approximate one, rests on a foundation of clearly defined requirements and is amenable to formal methods and analytical validation.

A number of impossibility results consider distributed consensus in general (see~\cite{DBLP:journals/jacm/FeketeLMS93,DBLP:journals/jacm/FischerLP85,DBLP:journals/dc/FischerLM86}). In~\cite{DBLP:books/mk/Lynch96}, the author shows that the presence of communication failures makes impossible to deterministically reach consensus (Theorem 5.1) and any $r$-round algorithm has probability of disagreement of at least $\frac{1}{r+1}$ (Theorem 5.5). This implies that there are no guarantees that vehicles can reach consensus on bounded time since vehicle-to-vehicle communications are prone to failures. Moreover, when the communication failures are too frequent and severe, vehicles can fail to reach  consensus for an unbounded number of consecutive times. We therefore abandon consensus-based decision algorithms, and prefer to focus on solutions that offer early  fall-back strategies against the risk of having disagreement for more than a constant number of rounds.

The existing literature on distributed (uniform) consensus algorithms with real-time requirements does consider processor failures. However, it often assumes timed and reliable communication. For example, in~\cite{DBLP:journals/tc/HermantL02} the authors give an algorithm that reaches agreement in the worst case  in time that is sublinear in the number of processors and  maximum  message delay. In~\cite{DBLP:conf/wdag/AguileraLT02}, the authors provide a time optimal consensus algorithm that reaches consensus in time $O(D(f+1))$ in the worst case where $D$ is the maximum  message delay and $f$ the maximum number of processors that can crash. In this paper, we do not assume reliable communication. Thus, message drops can occur independently among processors at any time.

Group communication systems~\cite{DBLP:journals/csur/ChocklerKV01} treat a group of participants as a single communication endpoint. The group membership service~\cite{DBLP:journals/tmc/DolevSW06,DBLP:journals/acta/DolevS04,DBLP:journals/tpds/DolevS03} monitors the set of recently live and connected participating system components whereas the multicast service delivers messages to that group under some delivery guarantees, such as delivery acknowledgment. In this paper we assume the existence of a membership service and a best-effort (single round solution) dissemination (multicast) protocol that has no delivery acknowledgment.

There exists literature on adaptive cruise control~\cite{stankovic2000decentralized, zhang1999using} as well as vehicle platooning~\cite{shladover1989longitudinal,shladover1991automated}. In~\cite{DBLP:conf/vnc/Lann11}, the author considers vehicle platooning and lane merging, and  bases his construction on distributed high level communication primitives. We consider a different failure model for which there is no deterministic implementation for these communication primitives.

The studied problem is motivated by the KARYON project~\cite{DBLP:conf/sss/CasimiroKKSTCPJL12,DBLP:conf/icdcsw/CasimiroPPS14,DBLP:conf/safecomp/BergerPPS14}. The KARYON project aims to provide a predictable and safe coordination of intelligent vehicles that interact in inherently uncertain environments. It proposes the use of a safety kernel that enforces the service level that the vehicle can safely operate. A cooperative service level can ensure that vehicles follow the same performance level. In this paper, we study a communication protocol that  implements the KARYON's cooperative service level evaluator. In~\cite{CRP14}, we present the architecture that considers the interactions between the safety kernel, a local dynamic map and the cooperative service level evaluator. Unlike the earlier abstract presentation of the cooperative service level evaluator, this paper provides in detail, the design and analysis of the communication protocol.

\Subsection{Our contribution}
We study an elegant solution for cooperative vehicular systems that have to deal with communication uncertainties. We base the solution on a communication protocol that, we believe, can be well understood by designers of safety-critical, automated and cyber-physical systems. We explain how the designers of fault-tolerant cooperative applications can use this solution to deal with communication failures when uniformly deciding on a shared value, such as $s$.

We consider cooperative applications that must periodically decide on a shared values $s$. Since the consensus problem cannot be deterministically solved in the presence of communication failures, the system is doomed to disagreed on the value of $s$ (in the presence communication failures that are frequent and severe). We bound the period in which the vehicles can be unaware of such disagreements with respect to $s$. We prove and validate that this bound is no more than one communication round (in a vehicular system that deploys a single-hop network of wireless ad hoc communication). We also study the percentage of time during which the system avoids disagreement on $s$ using ns-3 simulations.

We exemplify how the proposed solution helps to guarantee safety. We consider vehicles that operate in a cooperative operational mode as long as they are aware that all the nearby vehicles are also in the same mode (with at most one communication round period of disagreement). However, if at least one vehicle is suspecting that another vehicle is not, all vehicles switch, within one communication round period, to a baseline operational mode so that the safety standards are met.

\Subsection{Document structure}
We list our assumptions and define the problem statement (Section~\ref{s:sys}), before providing the timed protocol for cooperation with disagreement correction (Section~\ref{s:dcp}) and its correctness proof (Section~\ref{s:c}). As protocol validation study, we consider computer simulation (Section~\ref{s:eva}). We discuss cooperative vehicular application (Section~\ref{s:cva}) and an example before the conclusions  (Section~\ref{s:con}).

\Section{System Settings}
\label{s:sys}
We consider a message passing system that includes a set $members$ of $n$ communicating prone-resilient vehicles. We refer to the vehicles with id $i$ as $p_i$. We assume that all vehicles have access to a common global clock with a sub-microsecond offsets by calling the function $clock()$. This could be implemented, for example, using global positioning systems (GPS)~\cite{allan1980accurate}. Hence, we assume that the maximum time difference along vehicles is at most $syncBound$. We consider that the system runs on top of a timed and fault-tolerant, yet unreliable, dissemination protocol, such as~\cite{DBLP:journals/tit/BoydGPS06,DBLP:journals/dc/GeorgiouGK11}, that uses $gossipSend_i(m)$ to broadcast message $m$ from vehicle $p_i \in members$ to all vehicles in $members$. We assume that end-to-end message delay is at most $messageDelay$ time. Thus, messages are either  delivered within $messageDelay$ time or omitted. The constant $messageDelay$ depends on distinct factors such as the MAC protocol that is used, vehicle velocity, interference, etc.  For example, this bound can be set to $100$ms  or less using, for example, dedicated short-range communications (DSRC)~\cite{camp2005vehicle}.

Vehicle $p_j$ receives $m$ from $i$ by raising the event $gossipReceive_j(i, m)$. We consider a fully connected  network topology. However, the network can arbitrarily decide to omit messages, but not to delay them for more than $messageDelay$ time. These assumptions allow the protocol to run in a synchronous round based fashion. We consider rounds of time $roundLength$ where $roundLength \geq messageDelay + 2syncBound$.

Every vehicle $p_i$ executes a program that is a sequence of {\em (atomic) steps}. An input event can  be either the receipt of a message or a periodic timer going off triggering $p_i$ to start a new iteration of the do forever loop.

We define the \emph{uncertainty period} as the period that vehicles can disagree. We say that there was a {\em communication failure} at round $r$ if there exists a vehicle that has not received the messages from all  vehicles during round $r$.

\Subsection{The task}
The system's task is to satisfy requirements~\ref{prp:dis1} to~\ref{prp:dis3}, which consider Definition~\ref{def:scp}.

\begin{definition}[Stable Communication  Period]
\label{def:scp}
A stable communication period $X[r_1, r_2]$ is the period of $r_2-r_1$ rounds in which the vehicles do not experience communication failures, i.e., all vehicle receive all messages during these periods. Otherwise,  it is called an unstable communication period, denoted by  $Y[r'_1, r'_2]$.
\end{definition}

We say that a stable communication period $X[r_1, r_2]$ is  maximal when rounds $r_1-1$ and $r_2+1$ are unstable communication periods. Analogously, we define a maximal unstable communication period $Y[r'_1, r'_2]$, see Figure~\ref{fig:be}. Thus, in any run, the communication may go through maximal stable and unstable periods (and then perhaps back to stable) for an unbounded number of times. Requirements~\ref{prp:dis1} to~\ref{prp:dis3} deal with what the system output at every vehicle should be when it goes between the different periods.

\begin{figure}[ht!]
    \centering
    \includegraphics[scale=1]{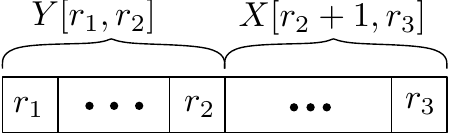}
    \caption{Maximal unstable communication period $Y[x_1, x_2]$ followed by a maximal stable communication period $X[x_2+1, x_3]$.}
    \label{fig:be}
  \end{figure}

\begin{property}[Certainty Period]
\label{prp:dis1}
During a stable period no two vehicles use different values. Moreover, within a bounded prefix of every stable period, there is a suffix during which no uses the default return value.
\end{property}

\begin{property}[Disagreement Correction]
\label{prp:dis2}
Every unstable period has a suffix named the {\em disagreement correction period} during which no two vehicles use different values. During this period all vehicles use the default return value. 
\end{property}

\begin{property}[Bounded Uncertainty Period]
\label{prp:dis3}
The suffix of a stable period during which some vehicles may use different values is called the {\em uncertainty period}. We require it to be bounded.
\end{property}

We show that any system run of the proposed solution fulfills requirements~\ref{prp:dis1} to~\ref{prp:dis3}. Specifically, we demonstrate Theorem~\ref{th:main} (Section~\ref{s:dcp}).

\begin{theorem}\label{th:main}
The proposed protocol (Algorithm~\ref{alg:TDDP}) fulfills requirements~\ref{prp:dis1},~\ref{prp:dis2}~and~\ref{prp:dis3}, where the uncertainty period is bounded by one round. Moreover, if vehicles do no experience communication failures, the disagreement correction holds for at most one round.
\end{theorem}

\Section{The Disagreement Correction Protocol}
\label{s:dcp}
We present the communication protocol in which the participants exchange messages until a deadline. These messages can include information, for example, about nearby vehicles as well as the confidences that each vehicle has about its information. Once everybody receives the needed information from each other, the participants can locally and deterministically decide on their actions. In case of a communication failure, each participant that experiences a failure imposes the default return value for one round.

Each vehicle $p_i \in members$ executes the protocol (that Algorithm~\ref{alg:TDDP} presents). It uses a do forever loop for implementing a round base solution. It accesses the global clock (line~\ref{ln:clock}) and checks whether it is time for the vehicle to send information about the current round (line~\ref{ln:roundClock1}). A vehicle starts sending  messages at $syncBound$ time from the beginning of each round and $syncBound+messageDelay$ before  of the end of each round using the $gossipSend()$ interface (Line~\ref{ln:gossipCall}). Recall that $syncBound$ is the maximum time difference over the vehicles and $messageDelay$ is the longest time that a message can live in the network. Next, it tests whether the current round number $myRound$ points to the current round in time (line~\ref{ln:roundClock}). A new round starts when  $myClock \div roundLength$ is greater than $myRound$.

At the beginning of every round, the protocol first keeps a copy of the collected data and the received information, and updates the round counter, as well as nullifying $data$ and $ack$ (line~\ref{ln:roundNull}). Then, it tests whether it has received all the needed information for the previous round (line~\ref{ln:arrivalTest}). Suppose that a communication failure occurred in the previous round, the protocol sets the data to be sent to the default return value $\bot$ (line~\ref{ln:assingment1}). It also writes to $controlLoop()$ interface the received information as well as the  default return value $\bot$ (line~\ref{ln:writeoutput1}). However, in the case that all messages of the previous round have arrived on time, the system reads the application information using $readState()$ interface. It also writes to $controlLoop()$ interface the received information  as well as the  value  that the deterministic function $decide()$ returns (line~\ref{ln:imposeEnd}).

The proposed protocol interfaces with the gossip (dissemination) protocol by sending messages ($gossipSend()$) and receiving them ($gossipReceive()$) periodically. The protocol locally stores the arriving information from $p_j \in members$ on each round in $data[k]$ and waits for the round end before it finishes accumulating all arriving information. More specifically, for each message that is reported with the same round, the protocol stores the data from $p_k$  and sets the acknowledgment variable  to true, if the message comes directly from $p_j$, ($k = j$), or transitively from $p_j$ without considering its own values ($ack_j[k]$ and $k \not=i$).

The correctness proof shows that, in the presence of a single communication failure, there could be at most one disagree round in which different system components use different values. Moreover, the influence of that single failure will last for at most two rounds, which is the shortest period possible. Note that Algorithm~\ref{alg:TDDP} handles well any sequence of communication failures.

\begin{algorithm*} [bt]
{\fontsize{9pt}{10pt}\selectfont
\begin{algorithmic}[1]
\STATE {\bf Constant:} {$members = \{p_1, p_2, ..., p_n\}$: The system vehicles.}

\STATE {\bf Constant:} $\bot$: Denotes a void (initialized) entry, as well as the default return value.
\STATE {\bf Constant:} $syncBound$: The maximum time difference among vehicle clocks.
\STATE {\bf Constant:} $maximumDelay$: The maximum time that a message time can live in the network.
\STATE {\bf Constant:} $roundLength >  2syncBound + maximumDelay$: The length of a round.
~~\\

\STATE {\bf Variable:} $myRound \gets 0$: Current communication round.
\STATE {\bf Variable:} $myClock \gets 0$: Current clock.
\STATE {\bf Variable:} $data[n] =\{\ldots \}$: Application data where $data[k]$ is the data received at round $myRound$ from member $p_k$.
\STATE {\bf Variable:} $ack[n] = \{false, \ldots \}$: Acknowledge for data reception where $ack[k]$ is true if $p_i$  has received (directly or indirectly)
the message from $p_k$ of the current round.
~~\\

\STATE {\bf Interface} 	$gossipSend()$: Disseminate information to the system members.
\STATE {\bf Interface} 	$gossipReceive()$: Dispatch arriving messages.
\STATE {\bf Interface} 	$readState()$: Return a datum to be sent.
\STATE {\bf Interface} 	$controlLoop()$: Write decided output.
\STATE {\bf Interface} 	$decide(s)$ Deterministically determines an item from $s$. We assume that whenever $\bot \in s$, then $\bot = decide(s)$. \label{al:fun}
~~\\
~~\\

\STATE {\bf Upon }$gossipReceive(j, <round_{j}, data_{j}, ack_j>)$ \label{ln:gossipReceive}
\IF {$(myRound = round_j)$}
	\FORALL{$p_k \in members$}		
		\IF {($ack_j[k]$ \AND $i \not= k$) \OR ($k = j$)}
			\STATE  $(data[k],ack[k]) \gets (data_j[k], true)$
		\ENDIF
	\ENDFOR
\ENDIF
\LOOP
	\STATE    $myClock \gets clock()$ \label{ln:clock}
	\IF{$myClock  \in [roundLength \cdot myRound+syncBound, roundLength\cdot (myRound +1)  - (syncBound + maximumDelay)]$}  \label{ln:roundClock1}

		\STATE   $gossipSend(i, <myRound, data, ack>)$\label{ln:gossipCall}
	\ENDIF
	\IF{$myRound < myClock \div roundLength$}  \label{ln:roundClock}
		\STATE $(s, r, myRound, data, ack[k]) \gets (data, ack, myClock \div roundLength, \{\bot, \ldots \}, k=i: \forall p_k \in members)$\label{ln:roundNull}
		\IF{$false \in \{r[k]: p_k \in members\}$} \label{ln:arrivalTest}
			\STATE $data[i] \gets \bot$ \label{ln:assingment1}
			\STATE $controlLoop(s, \bot)$	\label{ln:writeoutput1}
		\ELSE
			\STATE $data[i]  \gets readState()$\; \label{ln:assingment3}
			\STATE  $controlLoop(s, decide(s))$ \label{ln:imposeEnd}
		\ENDIF
	\ENDIF
\ENDLOOP
\end{algorithmic}
  \caption{Timed Protocol for Cooperation with Disagreement Correction (code for $p_i$)}
   \label{alg:TDDP}
}
\end{algorithm*}

\Section{Correctness} \label{s:c}
We prove that  Algorithm~\ref{alg:TDDP}  follows requirements~\ref{prp:dis1} to~\ref{prp:dis3}.	

%
%
%
%

\begin{lemma} \label{th:boundeduncertainty}
Let  $Y[r_1, r_2]$ be any maximal unstable communication period followed by a maximal stable communication period $X[r_2+1, r_3]$.
The following three statements hold.

(1) {\bf Bounded Uncertainty Period}. Vehicles may have disagreements at round $r_1+1$.

(2) {\bf Disagreement Correction}. All vehicles use the default return value during $[r_1+2, r_2+1]$.

(3) {\bf Certainty Period}. Vehicles use the same value during $[r_1+2, r_3+1]$.
\end{lemma}

\begin{proof}
Let $s_i(r)$ be the set of messages that vehicle $p_i$ receives from all the vehicles $p_j \in P$, either directly or indirectly, and that $p_j$ has sent during round $r$. Observe that each vehicle decides the value to be used on round $r+1$ based on the received information at round $r$ (lines~\ref{ln:writeoutput1}~and~\ref{ln:imposeEnd}). We claim that $s_k(r) = s_i(r)$ for $\forall p_k, p_\ell \in P$ and $\forall r \in [x_2+1, x_3]$. Note that this implies that no two vehicles use different values when processing round $r$, because vehicle $p_i$ determines its output value according to the deterministic function $decide(s_i)$.

\begin{claim}
$s_k(r) = s_\ell(r)$ for $\forall p_k, p_\ell \in P$ where $r \in [x_2+1, x_3]$.
\end{claim}

\begin{proof}[Claim Proof]
First we show that each vehicle maintains consistent its own information  over each round. Observe that lines between~\ref{ln:roundNull}~and~\ref{ln:imposeEnd}  are executed once during round $myRound$ since $myRound$  is set to  $clock() \div roundLength$ and $clock()$ always returns larger values. Therefore, each vehicle $p_i$ loads its message on the register $data[i]$ once during $myRound$. Thus, assume that vehicle $v_i$ overwrites its $data[i]$ when receiving a message from  vehicle $v_j$. Since the condition ensures that it loads $data[k]$ only if either $i\not=k$ or $k=j$, we conclude that $i=j=k$. Thus, $data[i]$ is consistent on $p_i$ during round $myRound$.

We say that a message $m_k$ is sent transitively, if $p_i$ receives $m_k$ from $p_j$ where $j\not=k$. We show that the message transitivity maintains the consistency of the messages during  a stable communication period. We argue by contradiction. Assume that there are two messages, $m_i \in s_k$ and $m'_i \in s_\ell$ such that $m_i\neq m'_i$. Consider the first time that $m_i,m'_i$ were sent. Observe that $p_i$ sent the two messages. A contradiction since $p_i$ maintains consistent its own information  over each round.

The claim follows by showing that at the end of the current round $myRound$, it holds that $s_k(myRound) = s_\ell(myRound)$. Indeed, since messages of each round are sent $(syncBound + maximumDelay)$ time units before the end of $myRound$ and after $syncBound$ time units after the beginning of $myRound$, vehicles receive messages only from the current round. Recall that $syncBound$ is the maximum difference time among vehicle clocks and  $maximumDelay$ is the maximum time that a message can live in the network.
\end{proof}

(1) {\bf Bounded Uncertainty Period}. Consider round $r_1$. Since  $Y[r_1, r_2]$ is an unstable communication period, there exists a vehicle $p_i$ that did not receive a message from all vehicles, i.e., $\bot$ is in $ack$. Observe that vehicle $p_j$ is unaware that $p_i$ had experienced a communication failure during round $r_1$. Let us assume that $p_j$ did not experience any communication failure. Therefore, $p_j$ uses the deterministic value that $decide()$ returns on round $r_1+1$. However, $p_i$ imposes the default return value (line~\ref{ln:writeoutput1}) since it had experienced a communication failure. Thus, during round $r_1+1$, $p_i$ sends the default return value by setting  $data[i]$ to $(\bot)$  and uses it (lines~\ref{ln:assingment1}~and ~\ref{ln:writeoutput1}, respectively). Therefore, as long as no vehicle misses $p_i$'s message, the first default return value of $p_i$ arrives along round $r_1 + 1$. Thus, during round $r_1+1$, $p_j$ uses a distinct value than $p_i$.

(2) {\bf Disagreement Correction}. We show that all vehicles use the default return value in  round $r \in [r_1+2, r_2+1]$. It is sufficient to show that there exists at least one default return value in  $s_j(r)$ in each round $r  \in  [r_1+1, r_2]$. Assume that at round $r \in [r_1, r_2]$, some vehicle $p_k$ experienced a communication failure. Therefore, at round $r+1$  all other vehicles either receive the default value of $p_k$  or  receive no message from $p_k$. Thus, all vehicles use the same value (default return value) in  each round $r \in [r_1+2, r_2+1]$ (lines~\ref{ln:assingment1}~and~\ref{ln:imposeEnd}). This is due to the definition of the function $decide()$ (line~\ref{al:fun}) and the fact that each vehicle writes the default return value if it experiences a communication failure.

(3) {\bf Certainty Period}.
We show that during  $[r_1+2, r_3+1]$ all vehicles use the same values. Indeed, from the point (2), every vehicle uses the default return value in every round $r \in [r_1+2, r_2+1]$. It remains to show that they use the same value in each round $r \in [r_2+2, r_3+1]$. From the claim, $s_i(r) = s_k(r)$ for each pair $p_i,p_k \in P$ during $[r_1, r_3]$  since all vehicles received the information from each other vehicle. The lemma follows since vehicles decide the value to be used on round $r+1$
based on the received information  at round $r$ (lines~\ref{ln:writeoutput1}~and~\ref{ln:imposeEnd}) using the deterministic function $decide()$.

\end{proof}

\begin{proof}[Theorem~\ref{th:main}]
It follows directly from Lemma \ref{th:boundeduncertainty}.
\end{proof}

\Section{Evaluation}\label{s:eva}
We consider a cooperative system that has two service levels where the lowest one is the default service level to which the system falls-back to in the presence of communication failures. For example, this can be a vehicular system in which the cooperative service level is the highest, and the autonomous service level is the lowest (default) one. Since we focus on communication failures, the experiments assume that every system component can always support the highest service level, and thus read input ($readState$) always returns the highest service level. We use computer simulation to validate the protocol as well as its efficiency. For the efficiency, we consider the {\em reliability} performance measure which we define as the percentage of communication rounds during which the protocol allows the system to run at its highest service level. First, we validate that the disagreement period is of at most one round and next the reliability of the protocol.

We simulate the  protocol using ns-3.~\footnote{http://www.nsnam.org/} We choose IEEE 802.11p as the communication channel with a log-distance path loss model and Nakagami fading channel model. Since DSRC technologies support end-to-end message delay of less than $100$ms~\cite{camp2005vehicle}, we fix the message delay to $100$ms. We consider a synchrony bound of $5$ms, say, using GPS~\cite{allan1980accurate} or a distributed clock-synchronization protocol. We implement a straightforward gossip protocol in which every node retransmits message every $50$ms.

We validate that the disagreement period is of at most one round. We plot in Figure~\ref{fig:be} the decision that $4$ vehicles took  independently during  $25$ rounds using the protocol under frequent communication failures. We set the round length to $160$ms so that messages can be transmitted twice in each round. Observe that at round $20$ vehicles $1$ and $2$ reduce the service level due to a communication failure,  but vehicles $3$ and $4$ still continue in the highest level of service. However, at round $21$, they lower their service level. Although vehicles do not operate on distinct service levels for more than one round, the service level of some vehicles may be oscillating. We can reduce this effect by increasing the round length. However,
the uncertainty period also increases.

\begin{figure}[!ht]
    \centering
    \includegraphics[scale=0.52]{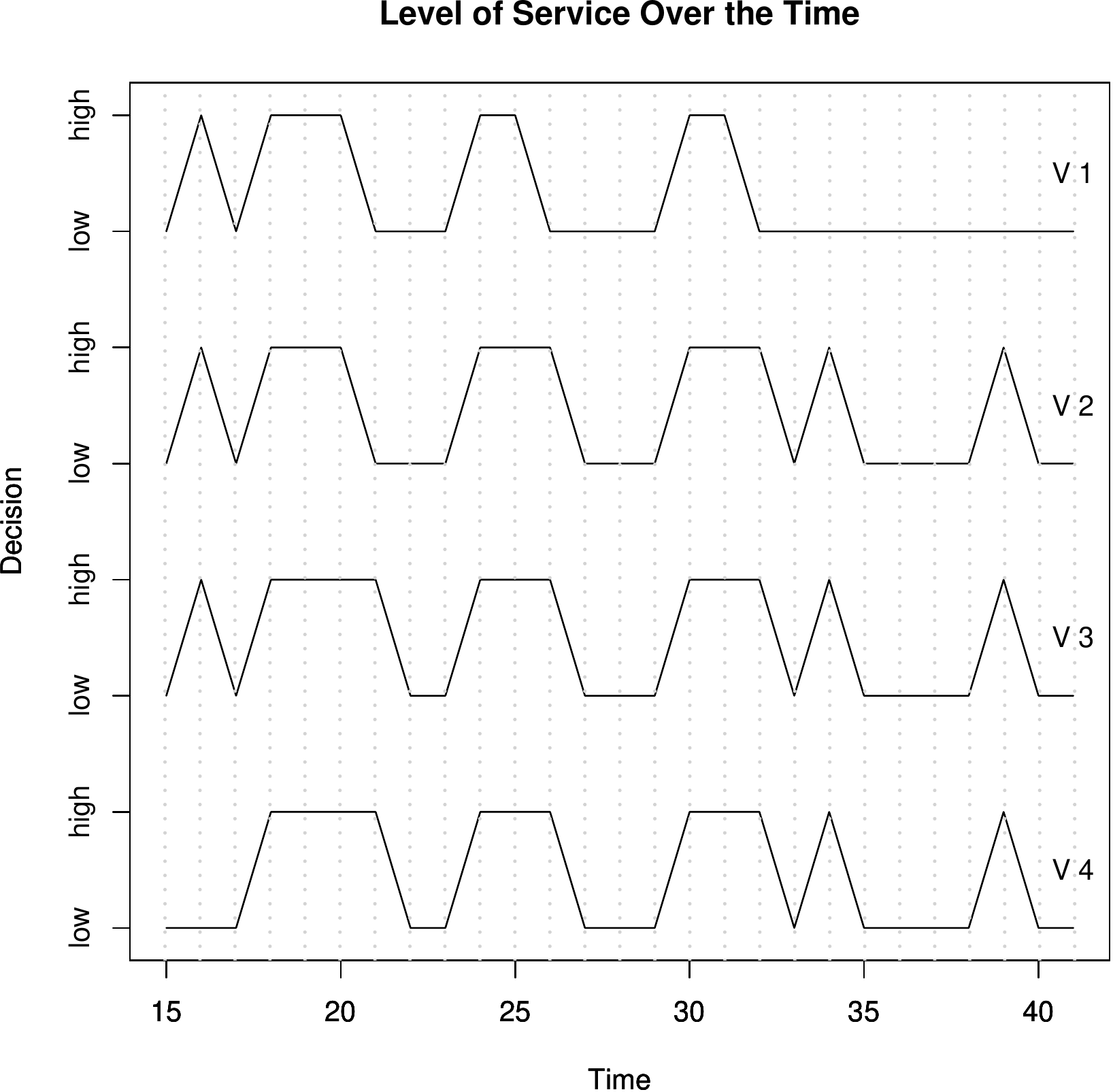}
       \caption{Vehicle behavior under frequent communication failures.  The plot shows the decision    that    four vehicles took among two service levels during $25$
        rounds using the protocol. }
    \label{fig:be}
  \end{figure}

Note the trade-off between the upper bound on the disagreement period, which is one communication round, and the success rate of the gossip protocol, which decreases as the round length becomes shorter. The type of gossip protocol as well as the number of system components also influences this success rate. We use computer simulation to study how these trade-offs work together and present the reliability.


We consider three round lengths  between $160$ms and $360$ms with intervals of $100$ms so that vehicles can transmit $2,4$ and $6$ messages in each round, respectively. We variate the number of vehicles  between two and eight. The reliability of the system is plotted in Figure~\ref{fig:rel}. We run each experiment for $360$ simulation seconds. During the simulations, we observe a packet drop average of $0.1436347$. The packet drop rate per number of vehicles is presented in  Table~\ref{fig:packetdrop}. Further, the  percentage of time that all vehicles agree on the highest service level is greater than $98\%$ with round lengths of at least $260$ms with at least four vehicle. Observe that the reliability is higher with more vehicles than with less. This is because of the transitivity property.


\begin{figure}[!ht]
    \centering
    \includegraphics[scale=0.57]{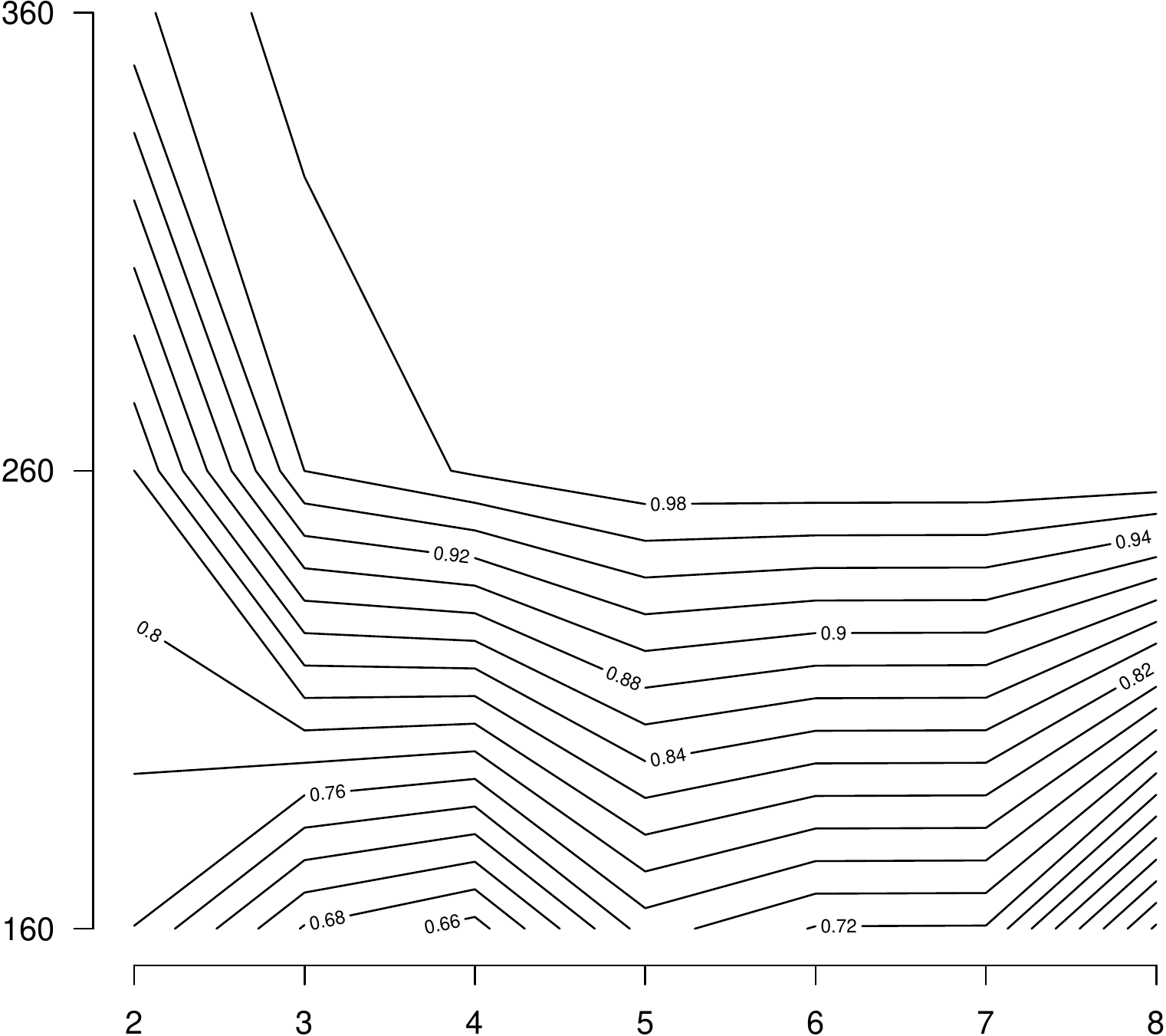}
     \caption{The percentage of time that all vehicles agree on the highest
    service level (number of vehicles vs the round length in milliseconds).}
    \label{fig:rel}
  \end{figure}

\begin{table}[ht!]
\centering
\begin{tabular}{|c|c|}
\hline
Number of Vehicles & Packet Drop Rate\\
\hline
2  & 0.1605357\\
3  & 0.1436347\\
4  & 0.159418 \\
5  & 0.141237 \\
6  & 0.1426173 \\
7  & 0.138037 \\
8  & 0.1713623 \\
\hline
\end{tabular}
\caption{Packet Drop Rate.}
\label{fig:packetdrop}
\end{table}

\Section{Discussion} 
\label{s:cva}
Autonomous vehicles have great capabilities to safely respond to unexpected events and keep short headways. However, keeping  short inter-vehicle distances without considering the nearby vehicles can result in hazardous situations. For example, rear-end crash as well as near-crash events  usually involve an action of the lead vehicle~\cite{lee2007analyses}. Cooperative vehicular systems have the potential to mitigate these events and improve the vehicle performance by exchanging information periodically as well as  their confidence level (\emph{validity}) about their own information. However, due to communication failures, vehicles may have inconsistent information and, therefore, low confidence level at the system level (even thought the individual vehicles may have high confidence). We present a cooperative vehicular application that exchanges periodically information and uses Algorithm~\ref{alg:TDDP} for dealing with communication failures. Our design demonstrates that, even though the presence of communication failures can lead to disagreement about what should be the joint validity value for a particular communication round, this can only happen for a period of at most one round, and thus tolerated by the vehicular control algorithm. We exemplify our approach using the adaptive cruise control (ACC) and vehicle platooning applications because their deployment environments requires dealing with communication uncertainties. In such cooperative systems, the vehicles has to jointly decide which application to use, ACC or platooning, according to the system service level $s$, which Algorithm~\ref{alg:TDDP} decides on its value.

We do not aim at designing new vehicular systems, but rather to exemplify how the proposed solution helps to guarantee the safety in existing cooperative vehicular applications, which operate in environments that include communication uncertainties. In our approach, while the vehicles are aware that nearby vehicles have a high level of certainty, they perform a fully cooperative operational mode to improve their performance. However, when this cannot be determined beyond any doubt, they switch to the autonomous operational mode to maintain high safety levels.

Adaptive Cruise Control (ACC) and Vehicular Platooning adjust the vehicle velocity so that they keep a predefined and safe headway. ACC sets the headway according to the vehicles in its direct line-of-sight. In other words, this application relies merely on on-board sensors. Vehicular platooning applications (or cooperative adaptive cruise control applications), however, do exchange information among vehicles and jointly aim at reducing air friction and energy consumption. They achieve such cooperative objectives by keeping shorter inter-vehicle distance than the autonomous ACC application. We show how to use the protocol solution for cooperative vehicle platooning with ACC as a base-line application.

In platooning, the vehicles exchange vectorial variables, $s$, which contain the vehicles' location, velocity and the highest level service that they can support. The service level provides the currently known bounds on the information error (see table~\ref{fig:saftyKernelParameters1}), as well as operational parameters, such as headway and acceleration bounds, where unbounded error means that the vehicle cannot determine it, say, due to a faulty component. Note that even thought our example considers merely three service levels, the extension to a scheme with more levels is straightforward. We assume that vehicles have the capability to determine the errors with high confidence level. We also assume that vehicles can determine the relative position of the vehicle ahead using on-board sensors within an estimated (bounded) error.

\begin{table}[ht!]
\centering
\begin{tabular}{|c|c|c|}
\hline
Level of Service & Loc. Err. ($P_\epsilon$) & Velocity Error ($S_\epsilon$)\\
\hline
High & $P_\epsilon \leq L$ & $S_\epsilon  \leq S$ \\
Medium & Unbounded $P_\epsilon$ & $S_\epsilon  \leq S$  \\
Low   & Unbounded $P_\epsilon$ & Unbounded $S_\epsilon$\\
\hline
\end{tabular}
\caption{$L$ and $S$ are constant values known by all participants.}
\label{fig:saftyKernelParameters1}
\end{table}

\begin{table}[ht!]
\centering
\begin{tabular}{|c|c|c|}
\hline
Level of Service & Headway & Acc Bound\\
\hline
High& $H_1$ & $\mathcal{A}_1$  \\
Medium & $H_2$ & $\mathcal{A}_2$ \\
Low   & $H_3$ & $\mathcal{A}_3$ \\
\hline
\end{tabular}
\caption{$H_{(\cdot)}$ are constant values and $\mathcal{A}_{(\cdot)}$ are constant acceleration bounds
such that $H_1 < H_2 < H_3$ and $\mathcal{A}_1 \subset \mathcal{A}_2 \subset \mathcal{A}_3$.}
\label{fig:saftyKernelParameters}
\end{table}

\begin{algorithm}[th!]
{\fontsize{9pt}{10pt}\selectfont
\begin{algorithmic}[1]
\STATE   {\bf Executes} the protocol that Algorithm~\ref{alg:TDDP} presents.

\STATE {\bf function} $decide(s)$
	\RETURN $\min_{\forall j \in s}(s[j].localLoS)$
~~\\
~~\\	

\STATE {\bf function} $readState()$
	\STATE Let $V$ and $localLoS$ be the operation information and maximum local service level that it supports, respectively, of $p_i$
	\RETURN $(localLoS, V)$
~~\\
~~\\

\STATE {\bf function} $controlLoop(data, LoS)$
	\IF{$p_i$ is the platoon leader}
	 	 \STATE Use $data$ and  acceleration bounds provided in Table~\ref{fig:saftyKernelParameters1} according to $LoS$ to maintain the cruise velocity if possible
	 \ELSE
	 	\STATE Use $data$ and  acceleration bounds provided in Table~\ref{fig:saftyKernelParameters1} according to $LoS$ to maintain the headway given in Table~\ref{fig:saftyKernelParameters}.
	 \ENDIF
	
\end{algorithmic}
\caption{Cooperative vehicle platooning with ACC as a base-line application (code for vehicle $p_i \in members$).}
\label{alg:platoonAdaptiveCruiseControl}
}
\end{algorithm}

Algorithm~\ref{alg:platoonAdaptiveCruiseControl} executes an instance of the protocol (that Algorithm~\ref{alg:TDDP} presents) and implement the interface functions $readState$, $controlLoop$ and $decide$. The function $readState$ returns the $p_i$'s local service level (see Table~\ref{fig:saftyKernelParameters1}) as well as the operational information (localization, heading, velocity, etc.). The function $decide$ returns the minimum local service level in the data structure $s$ so that all vehicles can meet the required constrains. The main functionality is implemented in $controlLoop$ function. It uses the information of all the vehicle in $data$ to determine the velocity and acceleration for the next round according to the cooperative service level using the parameters in Table~\ref{fig:saftyKernelParameters}. Note that in the baseline application, ACC, vehicles can base their decision on sensory information from onboard sources. We assume that each operation mode is proven to be safe provided that the information meets the requirements, i.e., the errors are within the bounds that are given in Table~\ref{fig:saftyKernelParameters}. For the worst case scenario, the behavior of the platoon can be influenced by vehicles that are not part of the platoon. This is because some events can cause cascade effects if they occur during the communication failures. We observe that the period during which the system switches from the highest service level to the lowest is a critical time. 

The safety provision in Algorithm~\ref{alg:platoonAdaptiveCruiseControl} depends directly on the mechanical constraints and the parameters' election. From the previous section, it is reasonable to consider rounds of length at least $260$ milliseconds. Thus, the headway can be determined from the round length and the error bounds on the information. We observe that Algorithm~\ref{alg:platoonAdaptiveCruiseControl} can reduce the collision risk by enforcing the vehicles to operate in a common service level that has been proven to be safe according to the information quality that is associated with that level.


\subsection*{Example}
As an illustrative example of the propose solution, we consider three vehicles in the following worst case scenarios of two implementations: (1) a vehicular platooning application that does use the proposed solution for for its back-off strategy in the presence of communication failures, and (2) an implementation of vehicular platooning that does run Algorithm~\ref{alg:platoonAdaptiveCruiseControl}.

\begin{figure}[!ht]
    \centering
    \includegraphics[scale=0.8]{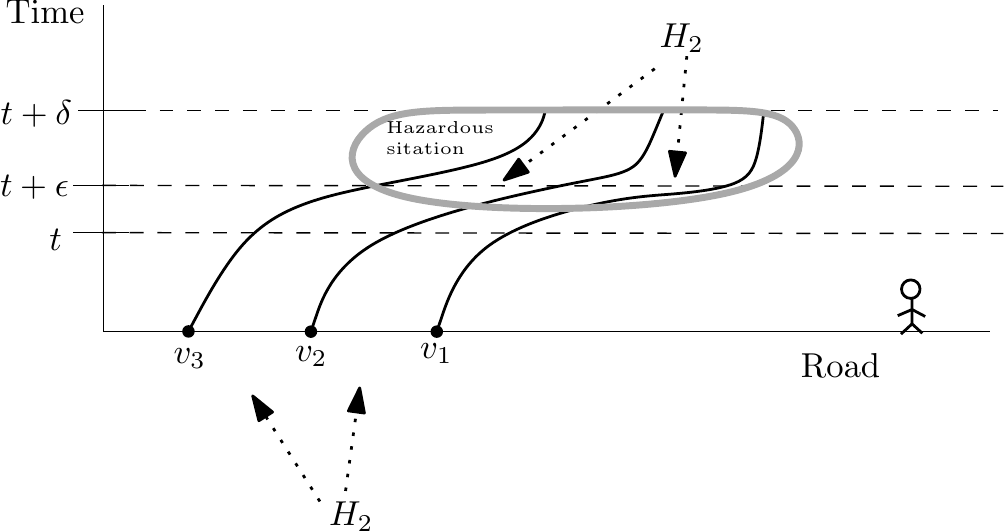}
    \caption{Vehicles in Platoon that do not include the proposed protocol.}
    \label{fig:exm1}
  \end{figure}

Let $v_1, v_2, v_3$ be the vehicles such that $v_1$ is leading the platoon followed by $v_2$ and $v_3$ as depicted in Figure~\ref{fig:exm1}. Assume that vehicles are driving on platooning with operational parameters given by the medium service level in Table~\ref{fig:saftyKernelParameters}. Therefore, they keep a headway of $H_2$. Suppose that at time $t$, $v_2$  starts loosing the messages from $v_3$ for $\delta$ time. Further, assume that at time $t+\epsilon$ , $v_2$ losses the messages from $v_1$  for $\delta - \epsilon$ time and at the same time $v_1$ requires to decelerate due to an obstacle, for example a pedestrian. Let us assume that $\epsilon$ and $\delta$ are at least two times the round length.

{\bf Platooning with back-off strategy that does not include the proposed solutions. }
Since $v_2$ does not receive the messages from $v_3$ during $[t, t+\delta]$, it is unaware whether $v_3$ continues operating on platooning. Thus, $v_2$ continues operating on platooning and assumes that it is the last vehicle in it. At time $t+\epsilon$, $v_2$ starts loosing messages from $v_1$ and consequently switches to the back-off strategy in the next round. However, since $v_1$ requires to brake, $v_2$ uses the acceleration bounds in $\mathcal{A}_3$. But $v_3$ continues operating on platooning during $[t, t+\delta]$, since it is unaware that $v_2$ is not receiving messages from $v_1$ and $v_3$. By definition, the system is not safe during $[t+\epsilon, t+\delta]$, since the platoon has only been proved to be safe when the headway is at most $H_2$ and acceleration bounds are in $\mathcal{A}_2$.

{\bf Platooning  using Algorithm~\ref{alg:platoonAdaptiveCruiseControl}.}
From the algorithm property that the uncertainty does not hold for more than one round,  $v_1$ and $v_3$ will be aware that at least one vehicle has a communication failure in the next round. Therefore, all switch to the lowest service level and start opening space to keep a headway of $H_3$. Thus, at time $t+\epsilon$ they have larger inter-vehicle distances which reduce the cascade effects. Observe that for an $\epsilon$ less than two round lengths, the problem  also occurs in this approach. Indeed, every cooperative vehicular application that relies on communication suffers from this problem. However, we believe that our approach minimizes the effects.

\begin{figure}[!ht]
    \centering
    \includegraphics[scale=0.8]{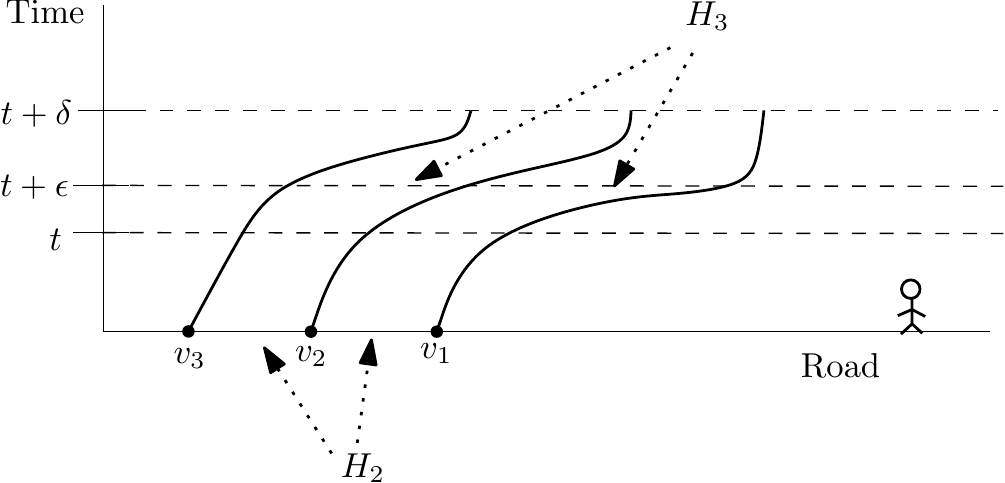}
    \caption{Vehicles in Platoon that include the proposed protocol.}
    \label{fig:exm2}
  \end{figure}

\Section{Conclusion} \label{s:con}
We have proposed an efficient protocol that can be used in safety-critical cooperative vehicular applications that have to deal with communication uncertainties. The protocol guarantees that all vehicles will not be exposed, for more than a constant time, to risks that are due to communication failures. We demonstrate correctness, evaluate performance and validate our results via ns-3 simulations. We also showed how vehicular platooning can use the protocol for maintaining system safety.

The proposed solution can be also extended to other cooperative vehicular applications, such as intersection crossing, coordinated lane change, as we demonstrated using the Gulliver test-bed~\cite{DBLP:conf/vtc/PahlavanPS12,DBLP:conf/mobiwac/PahlavanPS11} during the KARYON project~\cite{DBLP:conf/icdcsw/CasimiroPPS14}.~\footnote{Demonstration videos are available via \url{www.gulliver-testbed.net/documents}} Moreover, we have considered the simplest multi-hop communication primitive, i.e., gossip with constant retransmissions. However, that communication primitive can be substitute with a gossip protocol that facilitate a greater degree of fault-tolerance and better performance. This work opens the door for the algorithmic design and safety analysis of many cooperative applications that use different high-level communication primitives.

\end{document}